\documentclass[letterpaper, 10 pt, conference]{ieeeconf}  % Comment this line out if you need a4paper

\IEEEoverridecommandlockouts                              % This command is only needed if 
\usepackage{latexsym}
\usepackage[cmex10]{amsmath}
\usepackage{amssymb}
\usepackage{amsthm}
\usepackage{graphicx}
\usepackage{mathrsfs}
\usepackage{bbm}

\newcommand{\bA}{\mathbf{A}}
\newcommand{\bB}{\mathbf{B}}
\newcommand{\bC}{\mathbf{C}}
\newcommand{\bH}{\mathbf{H}}
\newcommand{\bI}{\mathbf{I}}

\newcommand{\bF}{\mathbf{F}}
\newcommand{\bG}{\mathbf{G}}

\newcommand{\Exp}{\mathsf{E}}
\newcommand{\f}{\mathsf{f}}
\newcommand{\g}{\mathsf{g}}
\newcommand{\h}{\mathsf{h}}

\newcommand{\bDe}{\mathbf{\Delta}}
\newcommand{\hS}{\hat{S}}

\newtheorem{theorem}{Theorem}
\newtheorem{lemma}{Lemma}
\newtheorem*{corollary}{Corollary}

\title{\LARGE \bf Adaptive Blind Separation of Two Dependent Sources$^*$
}

\author{George V. Moustakides$^{1}$, Feeby Salib$^{2}$ and Kalliopi Basioti$^{3}$% <-this % stops a space
\thanks{*This work was supported by the National Science Foundation under Grant CIF\,1513373, through Rutgers University.}% <-this % stops a space
\thanks{$^{1}$G.V. Moustakides is faculty with ECE, University of Patras, Rion, Greece, {\tt\footnotesize moustaki@upatras.gr} and with CS, Rutgers University, New Brunswick, NJ, USA, {\tt\footnotesize gm463@rutgers.edu}.}%
\thanks{$^{2}$F. Salib is graduate student with ECE, Rutgers University, New Brunswick, NJ, USA, {\tt\footnotesize fms50@scarletmail.rutgers.edu}.}%
\thanks{$^{3}$K. Basioti is graduate student with CS, Rutgers University, New Brunswick, NJ, USA, {\tt\footnotesize kib21@scarletmail.rutgers.edu}.}%
}%

\begin{document}

\maketitle
\thispagestyle{empty}
\pagestyle{empty}

%%%%%%%%%%%%%%%%%%%%%%%%%%%%%%%%%%%%%%%%%%%%%%%%%%%%%%%%%%%%%%%%%%%%%%%%%%%%%%%%
\begin{abstract}
We consider the problem of adaptive blind separation of two sources from their instantaneous mixtures. We focus on the case where the two sources are not necessarily independent. By analyzing a general form of adaptive algorithms we show that separation is possible not only for independent sources but also for sources that are dependent provided their joint pdf satisfies certain symmetry conditions. A very interesting problem consists in identifying the class of dependent sources that are non-separable, namely, the counterpart of Gaussian sources of the independent case. We corroborate our theoretical analysis with a number of simulations and give examples of dependent sources that can be easily separated.
\end{abstract}

%%%%%%%%%%%%%%%%%%%%%%%%%%%%%%%%%%%%%%%%%%%%%%%%%%%%%%%%%%%%%%%%%%%%%%%%%%%%%%%%
\section{Introduction and Background}
Blind source separation (BSS) is the problem of recovering
unobserved signals (sources) from their observed mixtures. BSS finds applications
in a number of areas as biomedical signal processing, speech
and image processing, data mining and communications \cite{c3}.

The simplest and most common version of the BSS problem
consists in estimating two source signals
$s_{1}(t),s_{2}(t)$ from two observations
$x_{1}(t),x_{2}(t)$ that are {\it
instantaneous} linear mixtures of the sources
of the form $x_{i}(t)=a_{i1}s_{1}(t)+a_{i2}s_{2}(t),~i=1,2$.
Using matrix notation, we can write
\begin{equation}
X_t=\bA S_t,
\end{equation}
where $X_t=[x_{1}(t),x_{2}(t)]^\intercal$ is the observation
vector, $S_t=[s_{1}(t),s_{2}(t)]^\intercal$ the source
signal vector and $\bA$ a constant matrix comprised of the
mixing coefficient $a_{ij},~i,j=1,2$. We assume that the observation
sequence $\{X_t\}$ becomes available sequentially and we are interested in the \textit{on-line} estimation of the source sequence $\{S_t\}$. It is clear that, to
solve this problem, it is sufficient to estimate the matrix
$\bB=\bA^{-1}$ since then $S_t$ can be recovered
as $\bB X_t$. Our results can be extended to cover multiple sources but we reserve the corresponding analysis for the more extended, journal version of our work.

For the solution of the BSS problem we concentrate on
adaptive algorithms. Therefore we will assume
that every time a new sample of the vector process $\{X_t\}$ becomes available we update an estimate $\bB_t$ of the inverse $\bA^{-1}$. We focus on adaptive algorithms of the form
\begin{align}
\begin{split}
\hS_t&=\bB_{t-1}X_t\\
\bB_t&=\bB_{t-1}-\mu\bH\big(\hS_t\big)\bB_{t-1},~\bB(0)=\bI,
\end{split}
\label{eq:B}
\end{align}
where (see \cite{c3,c4,c13}) the most common form of the matrix function $\bH(Z)$ is
\begin{equation}
\bH(Z)=[ZZ^\intercal-\bI]+[ZG^\intercal(Z)-G(Z)Z^\intercal],
\label{eq:ex_H}
\end{equation}
with $Z=[z_1\,z_2]^\intercal$, $G(Z)=[g_1(z_1)\,g_2(z_2)]^\intercal$,
$g_i(z)$ univariate functions,
$\bI$ the identity matrix, and $\mu>0$ is
a scalar {\it step size} that controls the convergence behavior of the algorithm.
Vector $\hS_t$ plays the role of the estimate of the source vector.
%$S_t$. %In fact the first term of $\bH(\cdot)$ makes the components of $\hat{S}_t$ uncorrelated while the second finally achieves the desired independence.

Although the algorithm defined by \eqref{eq:B}
is the one we apply in practice, for its analysis it is more convenient to adopt the following normalized version
\begin{align}
\begin{split}
\hS_t&=\bC_{t-1}S_t\\
\bC_t&=\bC_{t-1}-\mu\bH(\hS_t)\bC_{t-1},~\bC(0)=\bA,
\end{split}
\label{eq:mapa}
\end{align}
with $\bC_t=\bB_t\bA$ and where matrix
$\bA$ appears now only as initial condition.
Substituting the first equation in \eqref{eq:mapa} into the second yields the final recursion
\begin{equation}
\bC_t=\bC_{t-1}-\mu\bH(\bC_{t-1}S_t)\bC_{t-1},~\bC(0)=\bA,
\label{eq:calgor}
\end{equation}
which will be used in our subsequent analysis.

We will say that the adaptive algorithm solves the BSS
problem if $\bC_t$ tends {\it in the mean} to a {\it
non-mixing matrix} $\bC$ with the following possible forms
\begin{equation}
\bC=
\begin{bmatrix}
\pm c_1&0\\0&\pm c_2
\end{bmatrix},~\text{or}~
\bC=
\begin{bmatrix}
0&\pm c_1\\ \pm c_2&0
\end{bmatrix},
\label{eq:non-mix}
\end{equation}
where $c_1,c_2$ \textit{positive}, nonzero quantities. In other words $\bC$ must be either
diagonal or anti-diagonal with nonzero elements. These limits
impose an ambiguity in the ordering, power and sign of the estimated sources. Fortunately, in most applications these uncertainties can either be tolerated or corrected with simple means as, for example, employment of pilot signals, where periodically and at known time instances the source signals are synthetic and known before hand.

Remarkably, the algorithm defined in
\eqref{eq:calgor}, can converge even when very limited information about the statistical description of the sources is available. In fact, \cite{c4,c5} it is sufficient
that the functions $g_i(z)$ and the probability density
functions (pdf) of the sources satisfy certain symmetry
properties. We have the following theorem that summarizes
the existing results (for the two-source case).

\begin{theorem}\label{th:1}
Let the sources $\{s_{1}(t),s_{2}(t)\}$ satisfy the following assumptions:\\
A1. For every $t$, $s_{1}(t),s_{2}(t)$ are {\it independent}
random variables with {\it symmetric} densities and {\it at
most one} source can be Gaussian.

\noindent A2. For $\kappa_i=\Exp[g_i'(s_i)]\Exp[s_i^2]-\Exp[s_ig_i(s_i)],i=1,2,$  we have
$$
1+\kappa_1>0,~1+\kappa_2>0,~~(1+\kappa_1)(1+\kappa_2)>1.
$$

\noindent Then the adaptive scheme defined by (\ref{eq:calgor}) with $\bH(Z)$ defined in \eqref{eq:ex_H} can
converge in the mean to a non-mixing matrix and the
corresponding limit is locally stable.
\end{theorem}

\begin{proof} The proof can be found in \cite{c5}.\end{proof}

Theorem~1 does not guarantee global
convergence because of the nonlinear form of (\ref{eq:calgor}).
Worth mentioning is also the fact that in (\ref{eq:ex_H}) the first term in
$\bH(Z)$, which uses only second order moments, plays
the role of a whitener of the observation vector $X_t$, whereas the
second term, with the help of nonlinear statistics, imposes
the final independence and achieves separation. The literature on BSS is very rich. One can find a detailed review of the existing methodologies for the case of independent sources in \cite{c3}.

\section{Proposed Algorithmic Scheme}

%\subsection{Major Goals}
In this work, we extend the above result in two major directions. Specifically
\begin{itemize}
\item We show that there exists a rich
class of adaptive algorithms that can be applied to the BSS
problem with the same success as the algorithm defined in
\eqref{eq:B}, \eqref{eq:ex_H}. This algorithmic class not
only separates independent sources but also sources that are
{\it dependent}, provided that some simple symmetry
condition applies to the joint pdf of the source signals. It is in fact this symmetry that
guarantees separation and not independence. 

\item We identify the type of dependent random sources that {\it cannot
be separated} under our proposed general algorithmic scheme, hence
extending the non-Gaussianity requirement of the
independent case.
\end{itemize}

The motivation for considering dependent sources,
except of course the obvious theoretical
challenge, is the fact that even when sources are
independent under nominal conditions, once we consider
simple contamination models, independence can be easily lost.
For example if $\f(s_1,s_2)$ denotes the joint pdf of the two sources, the following $\epsilon$-contamination model does not correspond to independent sources
\begin{equation}
\f(s_1,s_2)=(1-\epsilon)\f_1(s_1)\f_2(s_2)+\epsilon \g_1(s_1)\g_2(s_2).
\label{eq:epscond}
\end{equation}
We see that with probability $1-\epsilon$
the two sources are independent following the pdfs $\f_1(s_1),\f_2(s_2)$ and with probability $\epsilon$ they are again independent but following the alternative pair of pdfs $\g_1(s_1),\g_2(s_1)$. It is a simple exercise to verify that $\f(s_1,s_2)$ does not correspond to independent sources. This raises the logical question as to whether the BSS algorithms will break under such mild divergence from the nominal conditions. There are of course applications \cite{c1,c2} where the source signals are genuinely dependent and we are interested in their separation. The existing literature for BSS methods for dependent sources is considerable. Here we only mention some representative articles for each available methodology. There are off-line techniques as Dependent Component Analysis \cite{c7,c8}, contrast functions \cite{c9}, time-frequency ratio of mixtures \cite{c10} and Kullback-Leibler divergence for copula densities \cite{c11} that are proposed to solve the problem. For on-line methods we find a technique based on nonnegative matrix factorization and the Kullback-Leibler divergence in \cite{c12}. For a more detailed list of references please consult: {\tt\small www.springeropen.com/collections/DCA}.

In this work, we adopt a purely algorithmic approach. Starting with the adaptive algorithm in \eqref{eq:B}, we examine what type of matrix functions $\bH(Z)$ can be employed and combined with what type of dependent sources in order for the algorithm in \eqref{eq:B} to be successful, namely, the algorithm in \eqref{eq:calgor} to converge to one of the non-mixing matrices. The goal is, whatever results we develop, to be applicable to a wide variety of signals without requiring exact knowledge of the statistical description of the sources.

%\section{Adaptive Separation of Dependent Sources}

Let us now introduce the adaptation we propose as a
general alternative to the existing algorithm in \eqref{eq:B} and \eqref{eq:ex_H}. Our
scheme also follows \eqref{eq:B} but with
the matrix function $\bH(Z)$ replaced by the more general
version
\begin{equation}
\bH(Z)=\left[\begin{array}{cc}
\h_{11}(z_1,z_2)&\h_{12}(z_1,z_2)\\
\h_{21}(z_1,z_2)&\h_{22}(z_1,z_2)
\end{array}\right].
\label{eq:new_H}
\end{equation}
For the analysis of the corresponding adaptive algorithm we will use the equivalent adaptation introduced in (\ref{eq:calgor}) with $\bH(Z)$ replaced by the expression introduced in \eqref{eq:new_H}. With our analysis we target the discovery of suitable constraints on $\bH(Z)$ that will guarantee the correct performance of the corresponding algorithm, namely its convergence to one of the non-mixing matrices depicted in \eqref{eq:non-mix}.

\section{Limits and Stability}
Adaptive algorithms can be analyzed using Stochastic Approximation theory \cite{c6} when the step size $\mu$ is ``small''. Our main interest lies with the convergence in the mean which we consider next.

\subsection{Limit in the Mean}
The mean field $\{\Exp[\bC_t]\}$ of the algorithm in \eqref{eq:calgor}, according to the Stochastic Approximation theory \cite{c13}, can be efficiently approximated by the sequence $\{\bar{\bC}_t\}$ defined by the recursion
\begin{equation}
\bar{\bC}_t=\bar{\bC}_{t-1}-\mu\Exp_S[\bH(\bar{\bC}_{t-1}S_t)]\bar{\bC}_{t-1}
\label{eq:meanfield}
\end{equation}
and the quality of the approximation is of the form
$$
\Exp[\bC_t]=\bar{\bC}_t+o(\sqrt{\mu}),
$$
where $\Exp_S[\cdot]$ denotes expectation only with respect to the source signal vector $S_t$. If we let $t\to\infty$ and assume that $\bar{\bC}_t\to\bar{\bC}_{\infty}$ with $\bar{\bC}_{\infty}$ being an invertible matrix we obtain the following equation for $\bar{\bC}_\infty$
\begin{equation}
\Exp_S[\bH(\bar{\bC}_{\infty}S_t)]=0.
\label{eq:limit}
\end{equation}
All matrices $\bar{\bC}_\infty$ that satisfy \eqref{eq:limit} are equilibrium points of the recursion in \eqref{eq:meanfield} and potential limits in the mean of the adaptive algorithm in \eqref{eq:calgor}. Whether a specific equilibrium can actually become the limit of the recursion in \eqref{eq:meanfield} is, of course, a question of stability of the particular equilibrium point. 

Let us ignore for the moment the stability issue and focus on the problem of \textit{imposing} a specific matrix as a possible equilibrium. We simply have to make sure that this matrix satisfies \eqref{eq:limit} when it replaces $\bar{\bC}_\infty$. 
%Regarding the desired non-mixing matrices defined in \eqref{eq:non-mix} we are interested in the diagonal case $\bC=\mathrm{diag}\{\pm c_1,\pm c_2\}$ and the anti-diagonal $\bC=\text{anti-diag}\{\pm c_1,\pm c_2\}$, with $c_1,c_2$ \textit{positive} quantities. 
To assure that the non-mixing matrices introduced in \eqref{eq:non-mix} are equilibrium points, we need for $i,j=1,2$ the following equations, corresponding to \eqref{eq:limit}, to be satisfied
\begin{equation}
\Exp[\h_{ij}(\pm c_1s_1,\pm c_2s_2)]=0
\label{eq:equat1}
\end{equation}
for the diagonal case, or
\begin{equation}
\Exp[\h_{ij}(\pm c_2s_2,\pm c_1s_1)]=0
\label{eq:equat2}
\end{equation}
for the anti-diagonal. We observe that, for simplicity, we have dropped the subscript ``$S$'' in the expectation $\Exp_S[\cdot]$ since, from now on, expectation is only with respect to the two sources.

We note that once the functions $\h_{ij}(z_1,z_2)$ are specified and the type of non-mixing matrix selected, then \eqref{eq:equat1} or \eqref{eq:equat2} constitutes a system of four equations in two unknowns ($c_1,c_2$). To have a solution of the desired form (diagonal or anti-diagonal) it is clear that two of the four equations must be satisfied \textit{automatically} and, most importantly, \textit{without exact knowledge} of the statistical description of the sources. It turns out that this is indeed possible if we assume that the joint pdf $\f(s_1,s_2)$ of the two sources exhibits the following property
\begin{equation}
\f(-s_1,s_2)=\f(s_1,-s_2)=\f(s_1,s_2),
\label{eq:quad_sym}
\end{equation}
corresponding to quadrantal symmetry. We should point out that \eqref{eq:quad_sym} is not an unfamiliar constraint. Indeed in the case of independent sources where $\f(s_1,s_2)=\f_1(s_1)\f_2(s_2)$ we recall from Theorem\,\ref{th:1}, Condition\,A1, that we need both marginal pdfs to be symmetric, which implies \eqref{eq:quad_sym}. Consequently we require the same symmetry to hold for the joint density when the two sources are dependent.

If we now impose some additional symmetries, this time on the functions $\h_{ij}(z_1,z_2)$, we can easily guarantee that the desired non-mixing matrices defined in \eqref{eq:non-mix} become equilibrium points. Specifically, we ask that the following conditions hold
\begin{align}
\begin{split}
&\h_{11}(-z_1,z_2)=\h_{11}(z_1,-z_2)=\h_{11}(z_1,z_2)\\
&\h_{22}(-z_1,z_2)=\h_{22}(z_1,-z_2)=\h_{22}(z_1,z_2)\\
&\h_{12}(-z_1,z_2)=\h_{12}(z_1,-z_2)=-\h_{12}(z_1,z_2)\\
&\h_{21}(-z_1,z_2)=\h_{21}(z_1,-z_2)=-\h_{21}(z_1,z_2).
\end{split}
\label{eq:mbafla1}
\end{align}
In other words, the two diagonal elements of the matrix $\bH(Z)$ must be \textit{even} functions in each of their arguments while the anti-diagonal \textit{odd} functions. There are two desirable consequences when these properties are combined with the quadrantal symmetry of the joint pdf $\f(s_1,s_2)$.
\begin{itemize}
\item For any $c_1,c_2$, we have $\Exp[\h_{12}(\pm c_1s_1,\pm c_2s_2)]=\Exp[\h_{12}(\pm c_2s_2,\pm c_1s_1)]=0$ and the same property is true for $\h_{21}(z_1,z_2)$. This suggests that two out of the four equations in \eqref{eq:equat1} or \eqref{eq:equat2} are satisfied for free and for all possible signs of the non-mixing matrix.

\item If $(c_1,c_2)$ are roots of the system of the two equations 
\begin{equation}
\Exp[\h_{11}(c_1s_1,c_2s_2)]=0,~\Exp[\h_{22}(c_1s_1,c_2s_2)]=0,
\label{eq:sys1}
\end{equation}
corresponding to a diagonal non-mixing matrix, or of the system 
\begin{equation}
\Exp[\h_{11}(c_2s_2,c_1s_1)]=0,~\Exp[\h_{22}(c_2s_2,c_1s_1)]=0,
\label{eq:sys2}
\end{equation}
corresponding to an anti-diagonal non-mixing matrix then so is any combination of signs $(\pm c_1,\pm c_2)$.
\end{itemize}
Both observations are very simple to demonstrate since they are a direct consequence of the symmetries imposed on $\h_{ij}(z_1,z_2)$ and $\f(s_1,s_2)$. Regarding $c_1,c_2$ we should point out that we only need their \textit{existence} since the exact values of these two quantities depend on the actual joint pdf $\f(s_1,s_2)$ which is assumed to be unknown. 

So far, through the symmetries imposed on $\f(s_1,s_2)$ in \eqref{eq:quad_sym} and on $\h_{ij}(z_1,z_2)$ in \eqref{eq:mbafla1}, we can guarantee that the desired non-mixing matrices are equilibrium points for the mean field adaptation \eqref{eq:meanfield}. However, in order for these equilibriums to be actually accessible as limits by the adaptation we also need to establish some form of stability.

\subsection{Local Stability}
The next step consists in examining under what conditions the non-mixing equilibrium points are in fact stable limits of \eqref{eq:meanfield}. We will present our analysis for $\bC=\mathrm{diag}\{c_1,c_2\}$, similar steps apply for the anti-diagonal case.

Establishing global stability in nonlinear updates is, unfortunately, very difficult and not always possible. We therefore limit ourselves (very common in adaptive algorithms) in testing only for \textit{local stability}. This means that we write $\bar{\bC}_t=\bC+\bDe_t$ where $\bDe_t$ is a perturbation matrix with ``small'' elements and analyze the evolution of $\bDe_t$ with $t$ using linear system approximation. Stability requires $\bDe_t\to0$ as $t\to\infty$. Specifically, assuming that the perturbation matrix is of the form
\begin{equation}
\bDe_t=\begin{bmatrix}\alpha_t&\gamma_t\\ \delta_t&\beta_t\end{bmatrix}
\label{eq:Delta}
\end{equation}
we have the following lemma that captures the evolution of $\bDe_t$.

\begin{lemma}\label{lem:0}
The elements of the perturbation matrix $\bDe_t$ satisfy the following recursions
\begin{align}
\begin{split}
&\begin{bmatrix}\alpha_t\\ \beta_t\end{bmatrix}=\bC(\bI+\mu\bF)\bC^{-1}\begin{bmatrix}\alpha_{t-1}\\ \beta_{t-1}\end{bmatrix}\\
&\begin{bmatrix}\gamma_t\\ \delta_t\end{bmatrix}=\bC(\bI+\mu\bG)\bC^{-1}\begin{bmatrix}\gamma_{t-1}\\ \delta_{t-1}\end{bmatrix},
\end{split}
\label{eq:mbofla9}
\end{align}
where 
\begin{align}
\!\!\bF&=\textstyle\Exp\left[\begin{bmatrix}\h_{11}(c_1s_1,c_2s_2)\\\h_{22}(c_1s_1,c_2s_2)\end{bmatrix}\!\! 
\left[\frac{s_1\f_{s_1}(s_1,s_2)}{\f(s_1,s_2)}~\frac{s_2\f_{s_2}(s_1,s_2)}{\f(s_1,s_2)}\right]\right]
\label{eq:mbofla101}\\
\!\!\bG&=\textstyle\Exp\left[\begin{bmatrix}\h_{12}(c_1s_1,c_2s_2)\\\h_{21}(c_1s_1,c_2s_2)\end{bmatrix}\!\! 
\left[\frac{s_2\f_{s_1}(s_1,s_2)}{\f(s_1,s_2)}~\frac{s_1\f_{s_2}(s_1,s_2)}{\f(s_1,s_2)}\right]\right],~
\label{eq:mbofla102}
\end{align}
$\f_{s_i}(s_1,s_2)=\frac{\partial \f(s_1,s_2)}{\partial s_i}$ and expectation in both formulas is with respect to the joint source pdf $\f(s_1,s_2)$.
\end{lemma}

\begin{proof} If we assume that $\f(s_1,s_2)$ is uniformly bounded then in order for its two marginal densities to exist we need $\lim_{s_1\to\pm\infty}\f(s_1,\cdot)=\lim_{s_2\to\pm\infty}\f(\cdot,s_2)=0$. In fact we need to strengthen this property slightly so that the two expectations in \eqref{eq:mbofla101} and \eqref{eq:mbofla102} are bounded. In particular, for any fixed constants $c_1,c_2$ we require
\begin{align}
\begin{split}
&\lim_{s_1\to\pm\infty}\h_{ij}(c_1s_1,\cdot)s_1\f(s_1,\cdot)=0\\
&\lim_{s_2\to\pm\infty}\h_{ij}(\cdot,c_2s_2)s_2\f(\cdot,s_2)=0.
\end{split}
\label{eq:ouff}
\end{align}
Details of the proof are given in the Appendix.\end{proof}

From the recursions in \eqref{eq:mbofla9} we can find conditions that assure local stability of the desired equilibrium. The next lemma discusses exactly this point.

\begin{lemma}\label{lem:3}
The equilibrium point $\bC$ is locally stable if and only if the following inequalities hold
\begin{equation}
\mathrm{tr}\{\bF\}<0,\mathrm{det}\{\bF\}>0,~~~\mathrm{tr}\{\bG\}<0,\mathrm{det}\{\bG\}>0,
\label{eq:R-H}
\end{equation}
where $\mathrm{tr}\{\cdot\},\mathrm{det}\{\cdot\}$ denote trace and determinant respectively.
\end{lemma}
\begin{proof} For local stability we need the two matrices $\bI+\mu\bF$, $\bI+\mu\bG$ to have their eigenvalues in the interior of the unit circle. This can happen \textit{for all sufficiently small step sizes} $\mu>0$ if and only if the two matrices $\bF,\bG$ have eigenvalues with \textit{strictly} negative real parts. The two inequalities applied to each matrix correspond to the Routh-Hurwitz criterion that assures this fact.\end{proof}
%\vskip0.1cm

\noindent\textbf{Remark\,1:} From our local analysis we observe that the mean estimates, when they are close to the limit, converge to the equilibrium \textit{exponentially} fast in the form of $(\bI+\mu\bF)^t$ and $(\bI+\mu\bG)^t$. In other words we have an exponential rate of convergence which is proportional to $\mu$. Of course this is true, provided that the conditions of Lemma\,\ref{lem:3} apply. As we can see, a smaller $\mu$ reduces the convergence speed towards the desired limit.
\vskip0.1cm

\noindent\textbf{Remark\,2:} We devoted all our efforts to assure convergence in the mean of the algorithmic scheme in \eqref{eq:calgor} to the desired non-mixing equilibrium. However, mean convergence by itself cannot guarantee satisfactory estimates. It is equally important that the \textit{variance} of the corresponding estimates is small. Fortunately, regarding this point, Stochastic Approximation comes to our rescue. Specifically, it is known \cite{c6} that when the limit in the mean is stable the corresponding covariance matrix of the estimates in \eqref{eq:calgor}, at steady-state, is proportional to $\mu$. Actually, there are even formulas that can compute the steady-state covariance matrix up to a first order approximation in $\mu$. Since the step size is selected to be small this suggests that, at steady-state, our estimates will differ from the desired non-mixing matrix by a random amount that has small power. Decreasing $\mu$ provides better steady-state estimates but, as mentioned in the previous remark, results in longer convergence periods of the mean field toward its desired limit.

\section{Non-Separable Sources}
One of the main issues in BSS is to identify the type of sources that cannot be separated. When the two sources are independent it is well known that the only combination that is non-separable by any off- or on-line method is the case of two Gaussians. If we allow the sources to be dependent with a joint pdf satisfying the symmetry in \eqref{eq:quad_sym} then the class of sources that are non-separable may increase. Unfortunately, under dependency it is very difficult to develop results of the same generality as in the independent case. Consequently, in order to come up with something meaningful, we propose a more modest characterization of the non-separable sources which, we believe, is equally interesting.

\vspace{0.1cm}
\noindent\textbf{Definition:} \textit{Two sources will be called non-separable if there is no algorithm of the form of \eqref{eq:calgor} for which a non-mixing equilibrium point $\bC$ defined in \eqref{eq:non-mix} is stable.}
\vskip0.1cm
\noindent In other words, instead of referring to any on- or off-line method, we relate the separability property to our general algorithmic scheme. If our algorithm is unable to converge to a non-mixing matrix no matter which functions $\h_{ij}(z_1,z_2)$ we employ, then we regard the corresponding sources as non-separable.

The equilibrium point is unstable if at least one of the two matrices $\bF,\bG$ has at least one eigenvalue with positive or \textit{zero} real part. Clearly this fact must be shared by \textit{all} combinations of functions $\h_{ij}(z_1,z_2)$ with symmetries specified in \eqref{eq:mbafla1}. We have the following theorem that identifies the joint probability density of sources that are non-separable, according to our definition.

\begin{theorem}\label{th:2}
Two dependent sources $s_1,s_2$ are non-separable by any version of the algorithm in \eqref{eq:calgor} if and only if their joint pdf $\f(s_1,s_2)$ is of the following form
\begin{equation}
\f(s_1,s_2)=\omega(K_2s_1^2+K_1s_2^2),
\label{eq:th2-1}
\end{equation}
where $\omega(z)$ is a univariate function of $z$ and $K_1,K_2$ are positive constants.
\end{theorem}

\begin{proof} The proof is very interesting and requires several steps. All details are given in the Appendix.\end{proof}

\begin{figure}[h]
%\vskip-0.2cm
\centerline{\includegraphics[width=0.6\hsize]{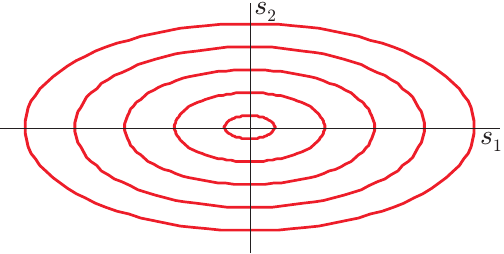}}
\vskip-0.2cm
\caption{Contour lines of the joint pdf of a pair of non-separable dependent sources.}
\label{fig:1}
\end{figure}
Theorem\,\ref{th:2} identifies as non-separable, the sources with joint pdf $\f(s_1,s_2)$ that exhibits \textit{elliptical} quadrantal symmetry (due to the term $K_1s_1^2+K_2s_2^2$). Fig.\,\ref{fig:1} captures the typical form of the contour lines of the corresponding joint pdf.
An interesting question is what happens when we apply our definition of non-separability
to the independent case. 
In particular, we would like to know whether our definition generates any additional, to the classical Gaussian pair, sources. The next corollary provides the necessary answer.

\begin{corollary}
When the two sources $s_1,s_2$ are independent the only combination which is non-separable by the algorithm in \eqref{eq:calgor} is the classical case of Gaussian sources.
\end{corollary}

\begin{proof} 
When the two sources are independent then $\f(s_1,s_2)=\f_1(s_1)\f_2(s_2)$. If we use this fact in \eqref{eq:th2-1} and take the derivative with respect to $s_1$ and $s_2$ we obtain the following equalities
\begin{align*}
\f'_1(s_1)\f_2(s_2)&=2K_2s_1\omega_z(K_2s_1^2+K_1s_2^2)\allowdisplaybreaks\\\allowdisplaybreaks
\f_1(s_1)\f'_2(s_2)&=2K_1s_2\omega_z(K_2s_1^2+K_1s_2^2),
\end{align*}
where $\omega_z(a)=\frac{d\omega(z)}{dz}|_{z=a}$. From the two equations we conclude that
$$
\textstyle\frac{\f'_1(s_1)\f_2(s_2)}{2K_2s_1}=\frac{\f_1(s_1)\f'_2(s_2)}{2K_1s_2}
$$
which suggests
$$
\textstyle\frac{\f_1'(s_1)}{\f_1(s_1)2K_2s_1}=\frac{\f_2'(s_2)}{\f_2(s_2)2K_1s_2}=K.
$$
$K$ must be a function solely of $s_1$ and at the same time a function solely of $s_2$, therefore it is necessarily a constant. The previous expression gives rise to two differential equations in $s_1$ and $s_2$, with solutions $\f_1(s_1)=A_1e^{KK_2s_1^2}$, $\f_2(s_2)=A_2e^{KK_1s_2^2}$, i.e.~Gaussian pdfs. The corresponding function $\omega(z)$ has the form $\omega(z)=A_1A_2 e^{Kz}$.
\end{proof}

The Corollary guarantees that, even if we limit ourselves to separation algorithms of the form of \eqref{eq:calgor}, this does not augment the class of non-separable sources when the sources are independent. This result was, in a sense, expected since from the literature we know that adaptive algorithms of the form of \eqref{eq:B}, with $\bH(Z)$ as in \eqref{eq:ex_H}, in simulation were seen to be able to separate independent sources except, of course, Gaussian pairs. Since our model for $\bH(Z)$ in \eqref{eq:new_H}, with the particular symmetries imposed in \eqref{eq:mbafla1}, is more general than \eqref{eq:ex_H}, the corresponding adaptive algorithm will also be capable of separating the same class of independent sources. Of course, the main value of Theorem\,\ref{th:2} comes from the fact that it identifies non-separable \textit{dependent} sources which is clearly not a straighforward extension of the Gaussian-pair of the independent case.

In the next section we give examples of classical and non-classical $\bH(Z)$ matrices and we test, using simulations, their capability to separate dependent sources. We also give examples of sources with elliptical quadrantal symmetry and verify that the algorithm in \eqref{eq:calgor} is unable to perform separation.

\section{Examples}
Let us start with the example where the pair $(s_1,s_2)$ is a mixture of independent Gaussian random variables. Specifically, with probability 0.5 the two sources $s_1$ and $s_2$ are independent $\mathcal{N}(0,1)$, $\mathcal{N}(0,4)$, while with probability 0.5 they are again independent $\mathcal{N}(0,4)$, $\mathcal{N}(0,1)$ respectively. We consider two cases for the $\bH(Z)$ matrix
\begin{align*}
\bH(Z)&=\begin{bmatrix}z_1^2-1&z_1z_2+z_1z_2^3-z_2z_1^3\\
z_1z_2+z_2z_1^3-z_1z_2^3&z_2^2-1\end{bmatrix}\\
\bH(Z)&=\begin{bmatrix}|z_1|-1&z_1z_2^2\textrm{sgn}(z_2)\\
z_2z_1^2\textrm{sgn}(z_1)&|z_2|-1\end{bmatrix}.
\end{align*}
The first matrix corresponds to the classical version introduced in \eqref{eq:ex_H} and, as we can see, it whitens the observations. On the other hand, the second matrix does not contain any whitening part. Both selections satisfy the symmetry properties set in \eqref{eq:mbafla1}. Furthermore, the analysis of the corresponding matrices $\bF,\bG$ assures validity of \eqref{eq:R-H} for stability.
\begin{figure}[!h]
\vskip-0.2cm
\centerline{\includegraphics[width=0.44\hsize]{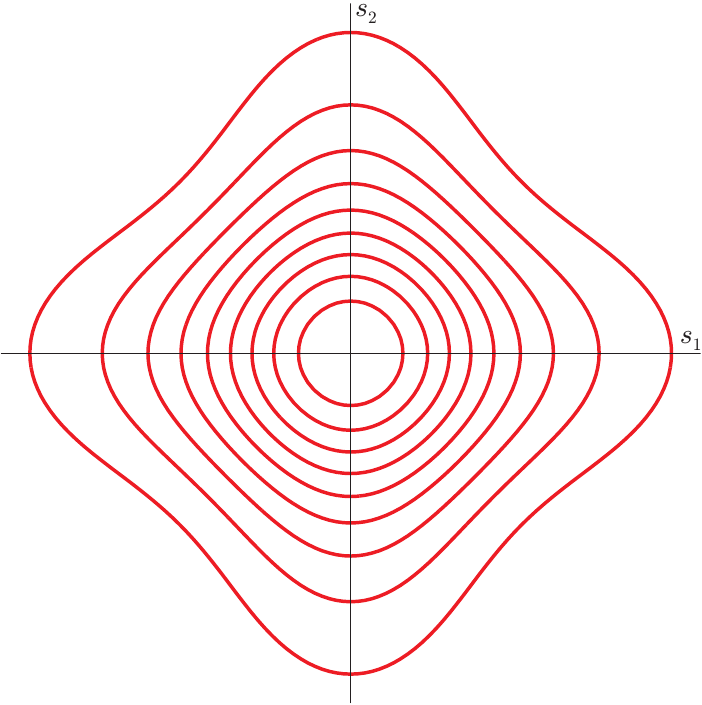}}
\vskip-0.1cm
\centerline{\footnotesize(a)}
\vskip0.3cm
\centerline{\includegraphics[width=0.48\hsize]{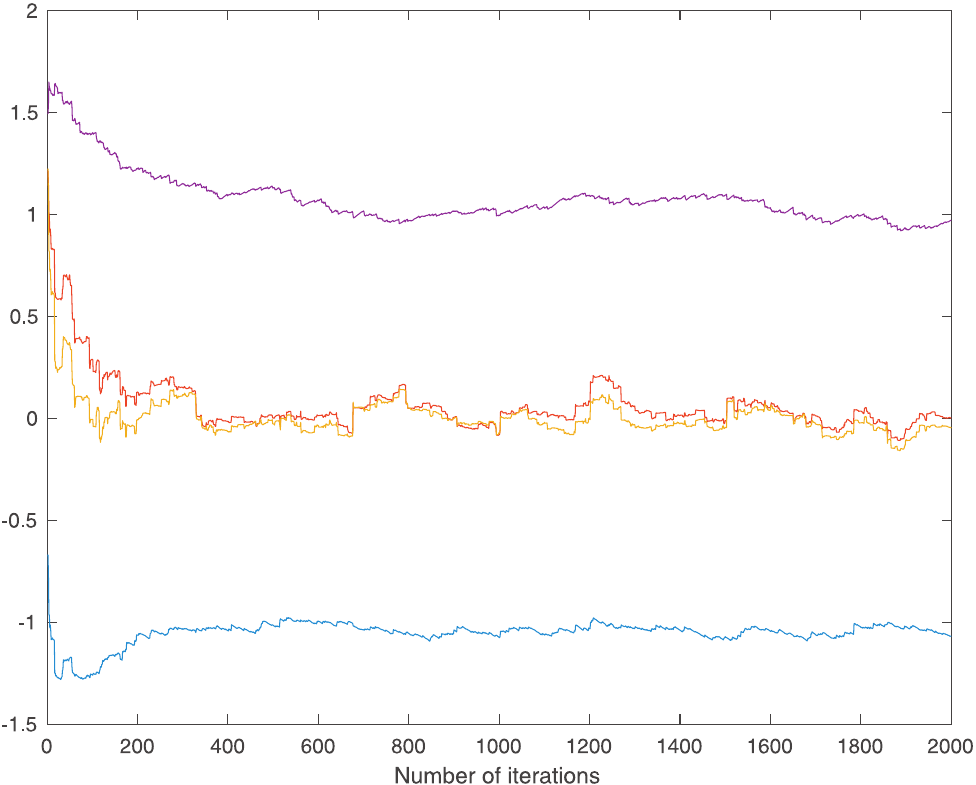}\hfill\includegraphics[width=0.48\hsize]{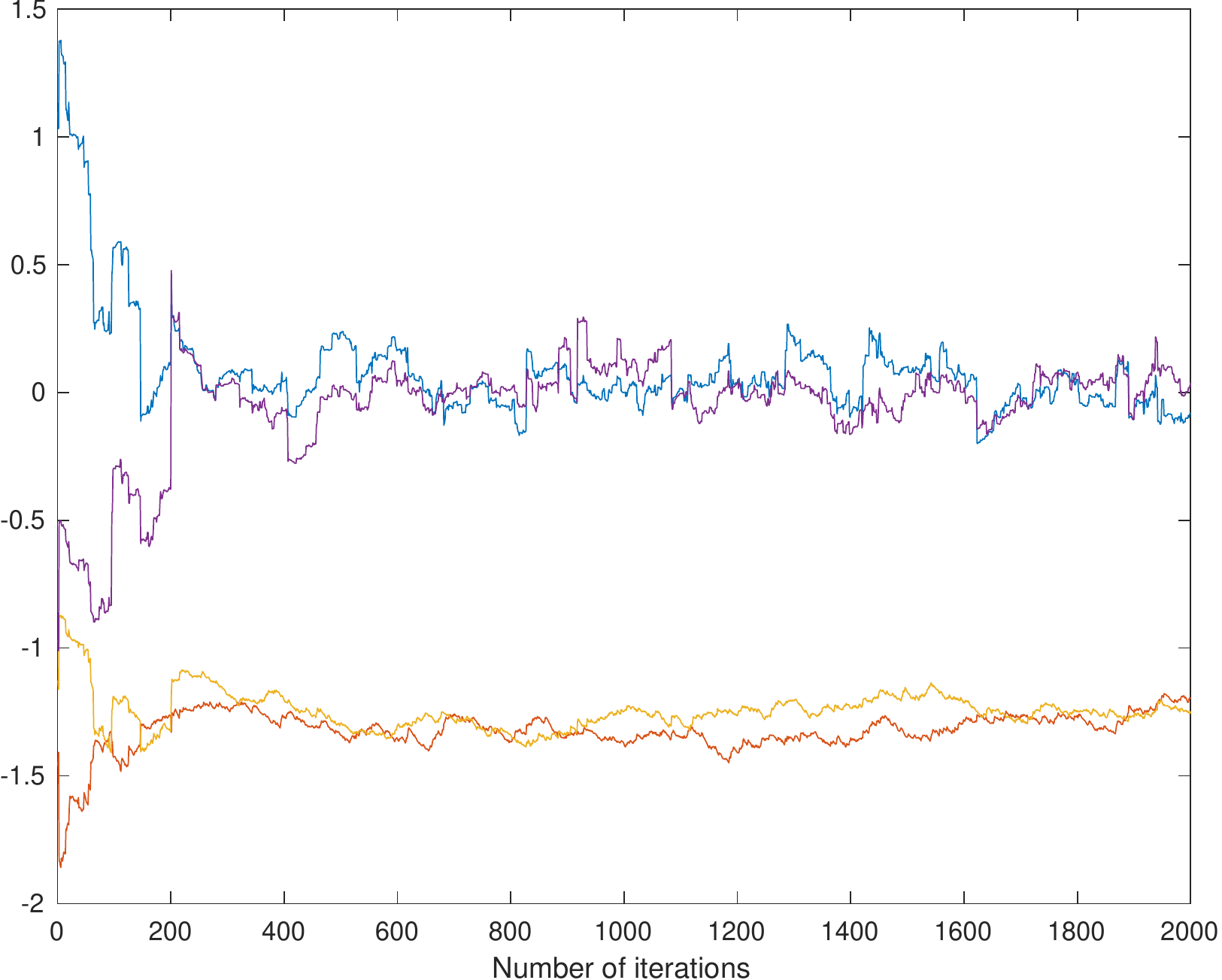}}
\vskip-0.1cm
\centerline{\hbox to 0.48\hsize{\hfill\footnotesize(b)\hfill}\hfill\hbox to 0.48\hsize{\hfill\footnotesize(c)\hfill}}
\vskip-0.1cm
\caption{(a) Contour lines of the joint pdf of two sources $s_1,s_2$ that are independent $\mathcal{N}(0,1),\mathcal{N}(0,4)$ with probability 0.5 and independent $\mathcal{N}(0,4),\mathcal{N}(0,1)$ with probability 0.5.  Evolution of the elements of $\bC_t$ with the number of iterations (b) classical, (c) without whitening.}
\label{fig:2}
\vskip-0.2cm
\end{figure}
 Fig.\,\ref{fig:2} depicts the simulation results. 
In Fig.\,\ref{fig:2}(a) we can see the contour lines of the corresponding joint pdf. We observe that we have quadrantal symmetry which is not elliptical, consequently the sources can be separated. In Fig.\,\ref{fig:2}(b) and (c) we plot the elements of the normalized estimates $\bC_t$ as they evolve in time for the two choices of $\bH(Z)$. We recall that $\bC_0=\bA$ where $\bA$ is unknown. We therefore initialized $\bC_0$ with a random matrix corresponding to a random selection of $\bA$ and $\bB_0=\bI$. Blue and magenta lines depict the diagonal elements of $\bC_t$ whereas yellow and orange the anti-diagonal. As we can see, in both algorithms we have convergence towards a non-mixing matrix.

Let us now test the validity of Theorem\,\ref{th:2}. We are going to generate dependent sources with their pdf controlled by a parameter $d$. When $d\neq0$ the joint pdf will have quadrantal symmetry but not elliptical. For $d=0$ the quadrantal symmetry will also become elliptical. This means that in the former case we expect source separability while in the latter the sources will be non-separable. The source model we propose is the following: We start with $r,\theta$ independent random variables with $r$ uniformly distributed in [0,1] and $\theta$ uniformly distributed in $[-\pi,\pi]$. We apply the following transformations to produce the two signals
\begin{equation}
s_1=r\cos\theta;~~~s_2=r\{\sin\theta+d(\sin\theta)^2\mathrm{sgn}(\sin\theta)\}.
\label{eq:dep_sources}
\end{equation}
As we mentioned, $d=0$ is the only value that generates elliptical (actually cyclic) symmetry. We use the classical $\bH(Z)$ matrix in order to demonstrate that the classical algorithms can also separate dependent sources. Fig.\,\ref{fig:3}(a) depicts the contour lines of the pdf and (b) the evolution of the elements of the normalized estimates for the case $d=1$. 
\begin{figure}[t]
\vskip0.3cm
\centerline{\hfill\includegraphics[width=0.4\hsize]{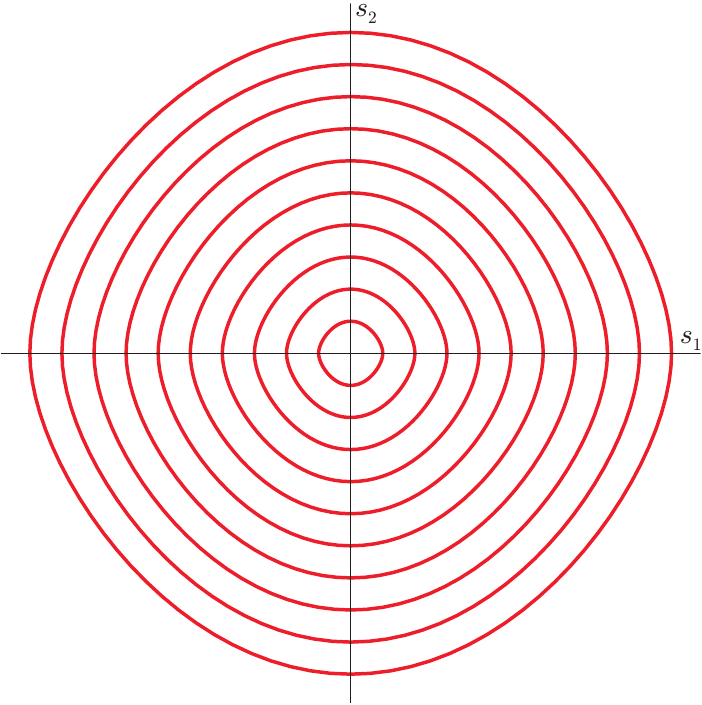}\hfill\includegraphics[width=0.5\hsize]{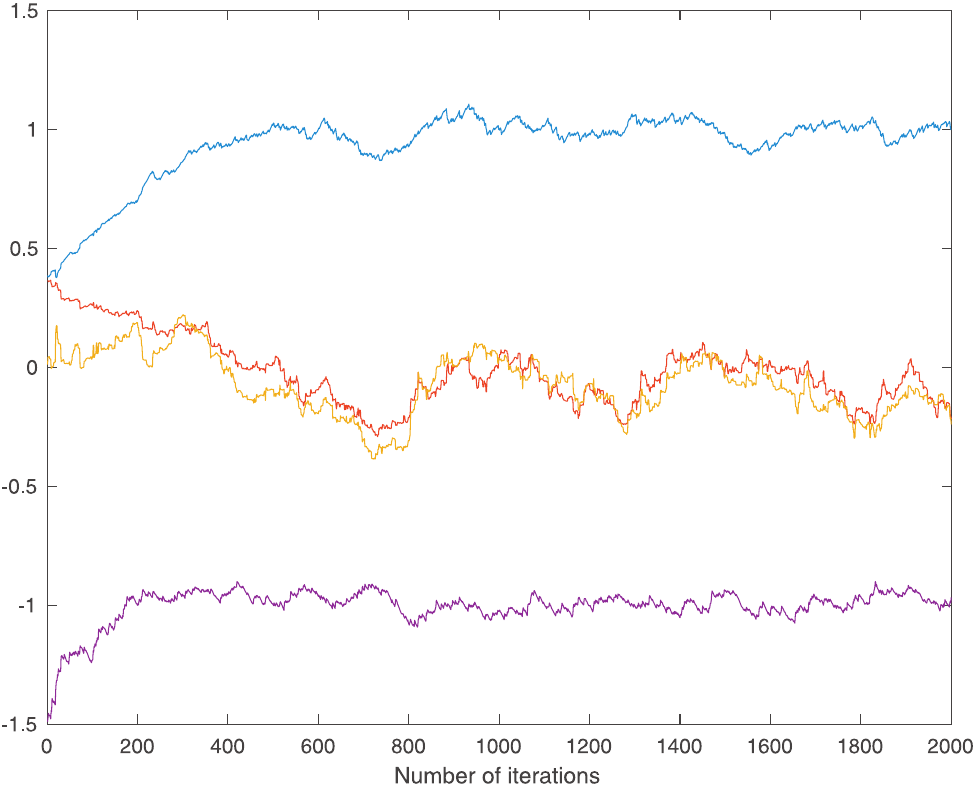}\hfill}
\vskip-0.1cm
\centerline{\hfill\hbox to 0.4\hsize{\hfil\footnotesize(a)\hfil}\hfill\hbox to 0.5\hsize{\hfil\footnotesize(b)\hfil}\hfill}
\vskip-0.2cm
\caption{(a) Contour lines of the joint pdf for sources following \eqref{eq:dep_sources} with $d=1$ and (b)~Evolution of the corresponding normalized estimates.}
\label{fig:3}
\vskip-0.5cm
\end{figure}
As we can see the adaptive algorithm converges to a non-mixing matrix.

In Fig.\,\ref{fig:4}, we present the simulation for $d=0$. Fig.\,\ref{fig:4}(a) has the contours which have indeed cyclic symmetry.
\begin{figure}[h]
\vskip0.2cm
\centerline{\hfill\includegraphics[width=0.4\hsize]{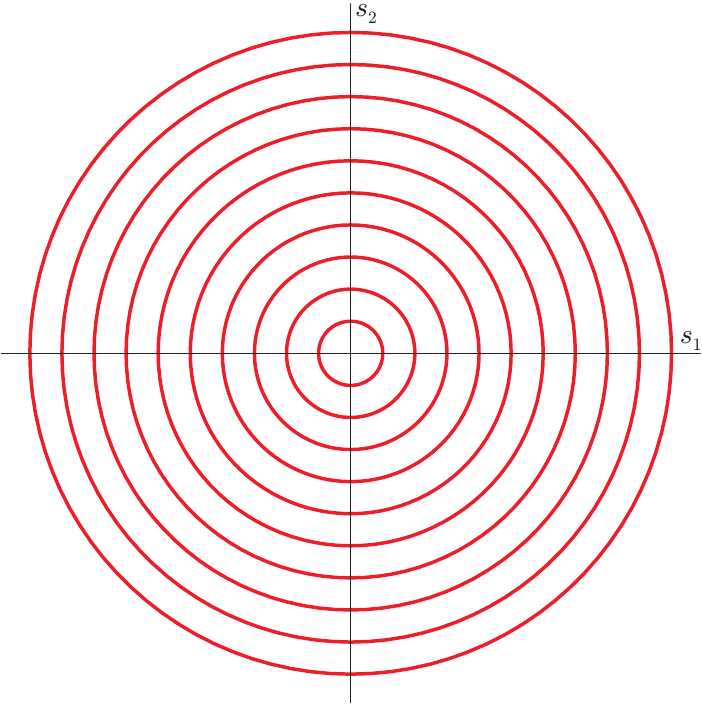}\hfill\includegraphics[width=0.5\hsize]{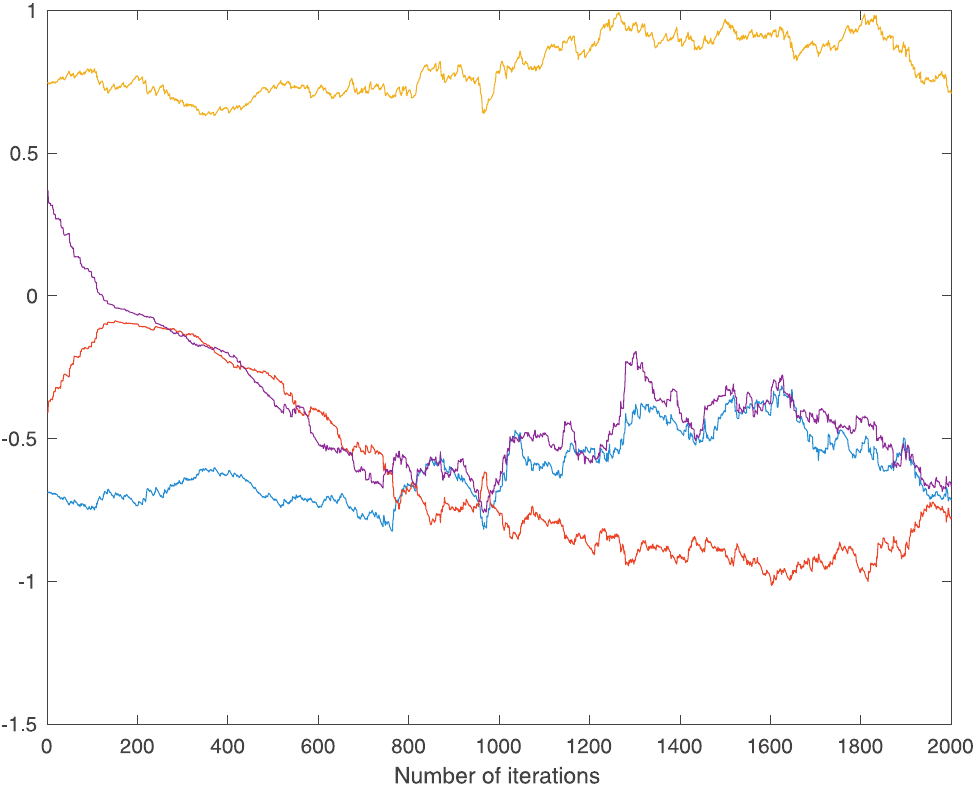}\hfill}
\vskip-0.1cm
\centerline{\hfill\hbox to 0.4\hsize{\hfil\footnotesize(a)\hfil}\hfill\hbox to 0.5\hsize{\hfil\footnotesize(b)\hfil}\hfill}
\vskip-0.2cm
\caption{(a) Contour lines of the joint pdf for sources following \eqref{eq:dep_sources} with $d=0$ and (b)~Evolution of the corresponding normalized estimates.}
\label{fig:4}
%\vskip-0.2cm
\end{figure}
In (b), we have the evolution of $\bC_t$ which, as predicted by our analysis, does not converge to a non-mixing matrix.

%Consequently these experiments corroborate our theoretical developments.


%%%%%%%%%%%%%%%%%%%%%%%%%%%%%%%%%%%%%%%%%%%%%%%%%%%%%%%%%%%%%%%%%%%%%%%%%%%%%%%%

%%%%%%%%%%%%%%%%%%%%%%%%%%%%%%%%%%%%%%%%%%%%%%%%%%%%%%%%%%%%%%%%%%%%%%%%%%%%%%%%

%%%%%%%%%%%%%%%%%%%%%%%%%%%%%%%%%%%%%%%%%%%%%%%%%%%%%%%%%%%%%%%%%%%%%%%%%%%%%%%%

\section{Conclusion}
We considered adaptive algorithms that are capable of blindly separating dependent sources. We showed that if the sources exhibit a quadrantal symmetry in their statistical behavior, then simple adaptive algorithms can be employed to separate them. This result indicates that source separability is not a property due to ``independence'' but rather due to ``symmetric statistical behavior''. With our analysis we were also able to identify the dependent sources that are not separable thus extending the Gaussian case known for independent sources.

\section*{Appendix}

\subsection*{Proof of Lemma\,\ref{lem:0}}
The proof is somewhat involved but presents no particular analytical challenges. Since $\bC$ is an equilibrium point satisfying $\Exp[\bH(\bC S)]=0$, where expectation is with respect to $\f(s_1,s_2)$, it is not very difficult to verify that the study of local stability of the mean field \eqref{eq:meanfield} at $\bC$ is the same as studying the local stability of
$$
\bar{\bC}_t=\bar{\bC}_{t-1}-\mu\Exp[\bH(\bar{\bC}_{t-1}S)]\bC
$$
at the same equilibrium. Consider now the perturbation $\bar{\bC}_t=\bC+\bDe_t$. Next we will present the complete computation of the recursion for the element $\alpha_t$ defined in \eqref{eq:Delta} for the perturbation matrix. Similar steps can be applied for the other three terms to show the validity of the lemma. Without loss of generality assume that $\bC=\mathrm{diag}\{c_1,c_2\}$, then for $\alpha_t$ we have
\vskip-0.8cm
\begin{multline*}
\alpha_t=\alpha_{t-1}-\\
~~~~~\mu c_1\Exp[\h_{11}(c_1s_1+\alpha_{t-1}s_1+\gamma_{t-1}s_2,c_2s_2+\delta_{t-1}s_1+\beta_{t-1}s_2)].
\end{multline*}
Applying first order Taylor expansion we obtain
\vskip-0.6cm
\begin{multline*}
\alpha_t=\alpha_{t-1}-\\
\mu c_1\big\{\Exp[\partial_{z_1}\h_{11}(c_1s_1,c_2s_2)s_1]\alpha_{t-1}+\\
\Exp[\partial_{z_1}\h_{11}(c_1s_1,c_2s_2)s_2]\gamma_{t-1}\\
+\Exp[\partial_{z_2}\h_{11}(c_1s_1,c_2s_2)s_1]\delta_{t-1}\\
+\Exp[\partial_{z_2}\h_{11}(c_1s_1,c_2s_2)s_2]\beta_{t-1}\big\},
\end{multline*}
where $\partial_{z_i}$ denotes partial derivative with respect to $z_i$. Since $\h_{11}(z_1,z_2)$ is even symmetric in $z_1$ this implies that $\partial_{z_1}\h_{11}(z_1,z_2)$ is odd in $z_1$ consequently $\partial_{z_1}\h_{11}(c_1s_1,c_2s_2)s_2$ is odd in both arguments $s_1,s_2$. Because $\f(s_1,s_2)$ exhibits quadrantal symmetry this suggests that $\Exp[\partial_{z_1}\h_{11}(c_1s_1,c_2s_2)s_2]=0$. Similar conclusion can be drawn for $\Exp[\partial_{z_2}\h_{11}(c_1s_1,c_2s_2)s_1]$. Because of this observation we can write
\vskip-0.6cm
\begin{multline*}
\alpha_t=\alpha_{t-1}-\\
\mu c_1\big\{\Exp[\partial_{z_1}\h_{11}(c_1s_1,c_2s_2)s_1]\alpha_{t-1}+\\
\Exp[\partial_{z_2}\h_{11}(c_1s_1,c_2s_2)s_2]\beta_{t-1}\big\}.
\end{multline*}
Let us now find a more convenient expression for the two expectations. First note that $\partial_{z_1}\h_{11}(c_1s_1,c_2s_2)=c_1^{-1}\partial_{s_1}\h_{11}(c_1s_1,c_2s_2)$. Using this equality we can write
\begin{multline*}
\Exp[\partial_{z_1}\h_{11}(c_1s_1,c_2s_2)s_1]=\allowdisplaybreaks\\ \allowdisplaybreaks\textstyle
c_1^{-1}\iint \partial_{s_1}\h_{11}(c_1s_1,c_2s_2)s_1\f(s_1,s_2)ds_1ds_2=\allowdisplaybreaks\\ \allowdisplaybreaks\textstyle
-c_1^{-1}\iint \h_{11}(c_1s_1,c_2s_2)\,\partial_{s_1}\!\big(s_1\f(s_1,s_2)\big)ds_1ds_2=\allowdisplaybreaks\\ \allowdisplaybreaks
\textstyle-c_1^{-1}\Exp\left[\h_{11}(c_1s_1,c_2s_2)\left(1+\frac{s_1\f_{s_1}(s_1,s_2)}{\f(s_1,s_2)}\right)\right]=\allowdisplaybreaks\\ \allowdisplaybreaks
\textstyle-c_1^{-1}\Exp\left[\h_{11}(c_1s_1,c_2s_2)\frac{s_1\f_{s_1}(s_1,s_2)}{\f(s_1,s_2)}\right].
\end{multline*}
For the second equality we used integration by parts and \eqref{eq:ouff} and for the last we used \eqref{eq:sys1}. Similar computations can be performed for the second expectation in the recursion for $\alpha_t$ and for the corresponding terms in the recursions for $\beta_t,\gamma_t,\delta_t$. This can prove validity of the formulas for the two matrices $\bF,\bG$ in \eqref{eq:mbofla101} and \eqref{eq:mbofla102}. This completes the proof of the lemma.\qed

\subsection*{Proof of Theorem\,\ref{th:2}}
As we mentioned, for non-separation we need at least one of the two matrices $\bF,\bG$ to have an eigenvalue which is either positive of zero. This property must be true for \textit{all} functions $\h_{ij}(z_1,z_2)$ with symmetries as in \eqref{eq:mbafla1}. Note that a possible selection of $\h_{ij}(z_1,z_2)$ is the following
\begin{gather*}
\textstyle\h_{11}(z_1,z_2)\!=\!-z_1\frac{\f_{s_1}(z_1,z_2)}{\f(z_1,z_2)},\,
\h_{22}(z_1,z_2)\!=\!-z_2\frac{\f_{s_2}(z_1,z_2)}{\f(z_1,z_2)}\\
\textstyle\h_{12}(z_1,z_2)\!=\!-\frac{z_2\f_{s_1}(z_1,z_2)}{\f(z_1,z_2)},\,
\h_{21}(z_1,z_2)\!=\!-\frac{z_1\f_{s_2}(z_1,z_2)}{\f(z_1,z_2)},
\end{gather*}
which satisfies the system of equations \eqref{eq:equat1} with $c_1=c_2=1$. Denote the corresponding $\bF,\bG$ matrices as $\bF_*,\bG_*$ then, using \eqref{eq:mbofla101}, \eqref{eq:mbofla102} we obtain
\begin{align*}
\bF_*&=\textstyle-\Exp\left[\begin{bmatrix}\frac{s_1\f_{s_1}(s_1,s_2)}{\f(s_1,s_2)}\\ \frac{s_2\f_{s_2}(s_1,s_2)}{\f(s_1,s_2)}\end{bmatrix} 
\left[\frac{s_1\f_{s_1}(s_1,s_2)}{\f(s_1,s_2)}~\frac{s_2\f_{s_2}(s_1,s_2)}{\f(s_1,s_2)}\right]\right]\\
%\end{equation*}
%\begin{equation*}
\bG_*&=
\textstyle-\Exp\left[\begin{bmatrix}\frac{s_2\f_{s_1}(s_1,s_2)}{\f(s_1,s_2)}\\ \frac{s_1\f_{s_2}(s_1,s_2)}{\f(s_1,s_2)}\end{bmatrix} 
\left[\frac{s_2\f_{s_1}(s_1,s_2)}{\f(s_1,s_2)}~\frac{s_1\f_{s_2}(s_1,s_2)}{\f(s_1,s_2)}\right]\right].
\end{align*}
Both matrices are \textit{symmetric and nonnegative definite}, therefore the only hope to experience instability is \textit{if and only if} at least one of the two matrices has an eigenvalue equal to 0 (since nonzero eigenvalues are necessarily negative). The latter can happen only when we can find constants $K_1,K_2$ such that $[K_1\,-\!\!K_2]^\intercal$ is an eigenvector to a 0 eigenvalue for $\bF_*$ or $\bG_*$. Because both matrices are symmetric and nonnegative definite, this is possible if and only if, at least one of the following two equations is satisfied for \textit{all} $(s_1,s_2)$
\begin{align}
&\textstyle
K_1s_1\f_{s_1}(s_1,s_2)-K_2s_2\f_{s_2}(s_1,s_2)=0
\label{eq:App-1}\\
&\textstyle
K_1s_2\f_{s_1}(s_1,s_2)-K_2s_1\f_{s_2}(s_1,s_2)=0,
\label{eq:App-2}
\end{align}
with not necessarily the same constants $K_1,K_2$.
Summarizing: For the specific selection of the $\h_{ij}(z_1,z_2)$ functions, at least one of the two equations \eqref{eq:App-1}, \eqref{eq:App-2} is required to be true if the sources are non-separable. 

It is a fact that: \textit{If
\eqref{eq:App-1} or \eqref{eq:App-2} is valid then for any other selection of $\h_{ij}(z_1,z_2)$ at least one of the corresponding $\bF$ or $\bG$ matrices will also have an eigenvalue equal to 0.}

This can be seen from \eqref{eq:mbofla101}, \eqref{eq:mbofla102} where we have the expression for $\bF,\bG$ for arbitrary $\h_{ij}(z_1,z_2)$. If for example \eqref{eq:App-2} is true then $\bG$ in \eqref{eq:mbofla102} will have the same $[K_1\,-\!\!K_2]^\intercal$ as a right eigenvector corresponding to a 0 eigenvalue. We can therefore conclude that if at least one of \eqref{eq:App-1}, \eqref{eq:App-2} is true then $\f(s_1,s_2)$ corresponds to non-separable sources.

Let us now examine what type of joint densities $\f(s_1,s_2)$ can satisfy \eqref{eq:App-1}, \eqref{eq:App-2}. 
We start with \eqref{eq:App-1}. Due to the quadrantal symmetry we can limit ourselves to the first quadrant with $s_1,s_2$ nonnegative. Define $z=s_1^{K_2}s_2^{K_1}$ then we can express $s_1$ in terms of $z$ and $s_2$ as $s_1=z^{\frac{1}{K_2}}s_2^{-\frac{K_1}{K_2}}$. Call $\omega(z,s_2)=\f(z^{\frac{1}{K_2}}s_2^{-\frac{K_1}{K_2}},s_2)$ and compute its partial derivative with respect to $s_2$, we have
\begin{multline*}
\textstyle
\omega_{s_2}(z,s_2)=-\frac{K_1}{K_2}z^{\frac{1}{K_2}}s_2^{-\frac{K_1}{K_2}-1}\f_{s_1}(z^{\frac{1}{K_2}}s_2^{-\frac{K_1}{K_2}},s_2)+\allowdisplaybreaks\\
\textstyle\f_{s_2}(z^{\frac{1}{K_2}}s_2^{-\frac{K_1}{K_2}},s_2)\!=\!-\frac{1}{K_2s_2}\Big\{ K_1z^{\frac{1}{K_2}}s_2^{-\frac{K_1}{K_2}}\f_{s_1}(z^{\frac{1}{K_2}}s_2^{-\frac{K_1}{K_2}},s_2)\allowdisplaybreaks\\
-K_2s_2\f_{s_2}(z^{\frac{1}{K_2}}s_2^{-\frac{K_1}{K_2}},s_2)\Big\}=0\allowdisplaybreaks
\end{multline*}
with the last equality coming from \eqref{eq:App-1}. From $\omega_{s_2}(z,s_2)=0$ we conclude that $\omega(z,s_2)=\omega(z)$. Recalling the relationship between $\omega(z,s_2)$ and $\f(s_1,s_2)$ and replacing $z$ with its definition we prove that $\f(s_1,s_2)=\omega(|s_1|^{K_2}|s_2|^{K_1})$. It turns out that functions of this form cannot be legitimate joint pdfs. This can be seen by integrating the equality over $s_2$ in order to identify the marginal pdf $\f_1(s_1)$. We note that
\begin{multline*}
\textstyle\allowdisplaybreaks
\f_1(s_1)=\int\f(s_1,s_2)ds_2=\int\omega(|s_1|^{K_2}|s_2|^{K_1})ds_2\allowdisplaybreaks\\\allowdisplaybreaks
\textstyle
=2\int_0^\infty\omega(|s_1|^{K_2}s_2^{K_1})ds_2\\=|s_1|^{-\frac{K_2}{K_1}}\frac{2}{K_1}\int_0^\infty z^{\frac{1}{K_1}-1}\omega(z)dz=|s_1|^{-\frac{K_2}{K_1}}A,
\end{multline*}
where constant $A$ is defined as $A=\frac{2}{K_1}\int_0^\infty z^{\frac{1}{K_1}-1}\omega(z)dz$.
The resulting form of $\f_1(s_1)$ is \textit{not} an integrable function over the whole real line for any value of the ratio $\frac{K_2}{K_1}$ and therefore cannot play the role of the marginal $\f_1(s_1)$. Consequently \eqref{eq:App-1} cannot be satisfied by any joint pdf $\f(s_1,s_2)$.

Let us now analyze in the same way \eqref{eq:App-2}. If in this case we define
$z=K_2s_1^2+K_1s_2^2$, solve for $s_1$ and follow exactly the same steps as in the previous case, we end up with $\f(s_1,s_2)=\omega(K_2s_1^2+K_1s_2^2)$. What is left to show is that $K_1,K_2$ must be of the same sign which, without loss of generality, can be considered positive. Note that if $K_2>0$ and $K_1<0$ then $K_2s_1^2+K_1s_2^2=r^2$, for fixed $r$, corresponds to a hyperbola. We can then express $s_1,s_2$ in terms of two alternative variables $r$ and $\theta$ as follows
%\newpage
$$\textstyle
s_1=\frac{1}{\sqrt{|K_2|}}r\cosh(\theta),~~s_2=\frac{1}{\sqrt{|K_1|}}r\sinh(\theta)
$$
where $r\geq0$ and $\theta$ can be any real. If the joint pdf of $s_1,s_2$ satisfies $\f(s_1,s_2)=\omega(K_2s_1^2+K_1s_2^2)$, then we can find the corresponding pdf of $r$ and $\theta$ by applying standard methodology for transformations of random variables, this yields
$$
\f(r,\theta)=\omega(r^2)r.
$$
The previous equation suggests that $r$ and $\theta$ are \textit{independent} and $r$ has a pdf of the form $A\omega(r^2)r$, where $A$ suitable constant, while $\theta$ a pdf equal to $A^{-1}$, namely, a \textit{uniform} density. The latter, however, is not possible since $\theta$ takes values on the whole real line and there is no uniform distribution that can support this unbounded range. Regarding this last point, one might argue that the density we are seeking exists and is known as ``degenerate uniform distribution''. However, we recall that this function is not an actual density but rather a \textit{limit} of a regular uniform density whose support increases without limit. The actual limiting function is \textit{not} a pdf since it is 0 everywhere on the real line. Consequently we cannot have $K_2>0$ and $K_1<0$ and the only legitimate choice for the joint pdf of a non-separable pair is a function with elliptical quadrantal symmetry. This completes the proof of the theorem.
\qed

%\addtolength{\textheight}{-17.75cm}

\end{document}